\documentclass[]{eptcs}
\providecommand{\event}{WPTE 2017} % Name of the event you are submitting to
\usepackage{breakurl}             % Not needed if you use pdflatex only.
\usepackage{latexsym}
\usepackage[pdftex]{graphicx}  % [dvipdfmx] for platex, which does not affect pdflatex.
\usepackage{latexsym}
\usepackage{amsmath}
\usepackage[psamsfonts]{amssymb}
\usepackage[all]{xy}
\usepackage{color}
\usepackage{thmtools, thm-restate}
\DeclareMathAlphabet{\mathbbold}{U}{bbold}{m}{n}
\usepackage[varg]{txfonts}

\newcommand*{\interpreted}[1]{{[\![ #1 ]\!]}}

\def\Epsilon{\mathcal{E}}

\def\rank{\mathrm{rank}}
\def\sig{\mathrm{sig}}
\def\DS{\mathcal{DS}}
\def\iff{\mathrel{\mathrm{iff}}}
\def\<#1>{\langle #1 \rangle}  % for \langle and \rangle by \< >
\newcommand{\vvec}[1]{{\bf #1}}
\newcommand{\wvec}[1]{\overrightarrow{#1}}

\newtheorem{theorem}{Theorem}%[section]
\newtheorem{definition}[theorem]{Definition}
\newtheorem{lemma}[theorem]{Lemma}
\newtheorem{proposition}[theorem]{Proposition}

\newtheorem{example}[theorem]{Example}
\newenvironment{proof}[1][Proof]{\begin{trivlist}
\item[\hskip \labelsep {\bfseries #1}]}{\end{trivlist}}

\providecommand{\squareforqed}{\hbox{\rlap{$\sqcap$}$\sqcup$}}
\providecommand{\QED}{{\unskip\nobreak\hfil%
	\penalty50\hskip1em\null\nobreak\hfil\squareforqed%
	\parfillskip=0pt\finalhyphendemerits=0\endgraf}}

\title{Reduced Dependency Spaces for Existential Parameterised Boolean Equation Systems\thanks{This paper is partially supported by JSPS KAKENHI Grant Number JP17H01721.}}
\author{Yutaro Nagae
\institute{\quad Graduate School of Information Science\\
\quad Nagoya University}
\email{nagae\_y@trs.cm.is.nagoya-u.ac.jp}
\and
Masahiko Sakai
\institute{\quad Graduate School of Informatics\\
\quad Nagoya University}
\email{\quad sakai@i.nagoya-u.ac.jp}
}
\def\titlerunning{Reduced DS for Existential PBESs}
\def\authorrunning{Nagae \& Sakai}
\begin{document}
\maketitle

\begin{abstract}
 A parameterised Boolean equation system (PBES) is a set of equations
 that defines sets satisfying the equations as the least and/or greatest fixed-points.
% This system is regarded as a declarative program defining functions that take a datum and return a Boolean value.
 Thus this system is regarded as a declarative program defining predicates,
 where a program execution returns whether a given ground atomic formula holds or not.
 The program execution corresponds to %This is
 the membership problem of PBESs,
% \MS{which decides} whether a given datum is in the defined set or not.
 which is however undecidable in general.

 This paper proposes a subclass of PBESs which expresses universal-quantifiers free formulas,
 and studies a technique to solve the problem on it.
 We use the fact that the membership problem is reduced to the problem whether a proof graph exists.
  To check the latter problem, we introduce a so-called dependency space which is a graph
 containing all of the minimal proof graphs.
 %A dependency space is a graph whose vertices are $X(v)$, which means the data $v$ is in the set $X$.
 %Each vertex also has type $\vee$/$\wedge$, which denotes if for some / for all vertex in the successors holds,
 %then it must hold.
 Dependency spaces are, however, infinite in general.
 Thus, we propose some conditions for equivalence relations
 to preserve the result of the membership problem,
 then we identify two vertices as the same under the relation.
 In this sense, dependency spaces possibly result in a finite graph.
% We show that a dependency space contains all of the proof graphs if they exist.
 We show some examples having infinite dependency spaces which are reducible to finite graphs by equivalence relations.
 We provide a procedure to construct finite dependency spaces
 and show the soundness of the procedure.
 We also implement the procedure using an SMT solver and experiment on some examples including a downsized McCarthy 91 function.
\end{abstract}

\section{Introduction}
A \emph{Parameterised Boolean Equation System} (PBES)~\cite{Groote2004a,Groote2004,GROOTE2005332} is a set of equations
denoting some sets as the least and/or greatest fixed-points.
PBESs can be used as a powerful tool for solving a variety of problems
such as process equivalences~\cite{Chen:2007:ECI:2392200.2392211}, 
model checking~\cite{Groote2004a,Groote:2005:MPD:1099102.1099103}, and so on.
 
We explain PBESs by an example PBES $\Epsilon_1$, which consists of the following two equations:
%\[
%  \left(\nu X(n:N)  =  X(n+1) \vee Y(n) \right)
%  \left(\mu Y(n:N)  =  Y(n+1)\right)
%\]
\[
 \begin{array}[t]{rcl}
  \nu X(n:N) & = & X(n+1) \vee Y(n)\\
  \mu Y(n:N) & = & Y(n+1)
 \end{array}
\]
$X(n:N)$ denotes that $n$ is a natural number and a formal parameter of $X$.
Each of the predicate variables $X$ and $Y$ represents a set of natural numbers regarding that $X(n)$ is true if and only if $n$ is in $X$.
These sets are determined by the equations,
where $\mu$ (resp.\ $\nu$) is a least (resp.\ greatest) fixed-point operator.
In the PBES $\Epsilon_1$, $Y$ is an empty set since $Y$ is the least set
satisfying that 
$Y(n) \iff Y(n+1)$ for any $n\geq 0$.
Similarly, $X$ is equal to $\mathbb{N}$ since $X$ is the greatest set
satisfying that 
$X(n) \iff X(n+1) \vee Y(n)$ for any $n\geq 0$.
A PBES is regarded as a declarative program defining predicates.
In this example, an execution of the program for an input $X(0)$ outputs true.

%\YN{
%The membership problem is regarded as an execution of the program.
%For instance, in this example, $X(2)$ means that this program takes
%some parameters $X$ and $2$, and returns $X(3) \vee Y(2)$,
%where $X(3)$ and $Y(2)$ are calculated recursively.
%}

% The membership problem for PBESs is undecidable in general, 
% since it is shown how a $\mu$-calculus formula and a process algebraic specification, both involving data
% parameters, can be transformed into a PBES~\cite{Groote2004a}.
The membership problem for PBESs is undecidable in general~\cite{Groote2004a}.
Undecidability is proved by a reduction of the model checking problem
for the modal $\mu$-calculus with data.
Some techniques
have been proposed to solve the problem for some subclasses of PBESs:
one by instantiating a PBES to a
\emph{Boolean Equation System} (BES)~\cite{PLOEGER2011637},
one by calculating invariants~\cite{invariants},
and one by constructing a proof graph~\cite{Cranen2013}.
In the last method, the membership problem is reduced to an existence of a proof graph.  
%A proof graph that justifies $X(0)$ for $\Epsilon_1$ is shown as follows:
%{\small
%  \[
%  \xymatrix{
%    X(0) \ar[r] & X(1) \ar[r] & X(2) \ar[r] & \cdots,
%  }
%  \]
%}
%where each vertex $X(n)$ represents that the predicate $X(n)$ holds.
If there exists a finite proof graph for a given instance of the problem,
it is not difficult to find it mechanically.
However, finite proof graphs do not always exist.
A technique is proposed in ~\cite{Nagae2016} that possibly produces a finite reduced proof graph, which represents an infinite proof graph.
The technique manages the disjunctive PBESs, in which data-quantifiers are not allowed.

In this paper, we propose a more general subclass, named \emph{existential PBESs}, and extend the notion of dependency spaces.
We discuss the relation between extended dependency spaces and the existence of proof graphs.
Dependency spaces are, however, infinite graphs in most cases.
Thus we reduce a dependency space for the existential class to a finite one 
in a more sophisticated way based on the existing technique in~\cite{Nagae2016}.
We also give a procedure to construct a reduced dependency space and show the soundness of the procedure.
We explain its implementation and an experiment on some examples.

\section{PBESs and Proof Graphs}
We follow~\cite{Cranen2013} and~\cite{GROOTE2005332} for basic notions related to PBESs and proof graphs.

We assume a set $\DS$ of \emph{data sorts}. 
For every data sort $D\in\DS$, we 
assume a set $\mathcal{V}_D$ of \emph{data variables} and a \emph{semantic domain} $\mathbb{D}$ corresponding to it.
%In this paper, we assume the existence of a sort $B$ corresponding to the Boolean domain $\mathbb{B} = \{\mathbbold{t},\mathbbold{f}\}$ and a sort $N$ corresponding to the natural numbers $\mathbb{N}$.
In this paper, we assume $B,N\in\DS$ corresponding to the Boolean domain $\mathbb{B} = \{\mathbbold{t},\mathbbold{f}\}$ and the natural numbers $\mathbb{N}$, respectively.
We use $D$ to represent a sort in $\DS$, $\mathbb{D}$ for the semantic domain corresponding to $D$, and $d$ and $e$ as a data variable in $\mathcal{V}_D$.
%A \emph{data environment} $\delta$ is a function 
%that maps each data variable to a value of the associated type.
%A \emph{data environment update} $\delta[v/d]$ for a data variable $d \in \mathcal{V}_D$ and $v\in\mathbb{D}$ is a mapping defined by $\delta[v/d](d') = v$ if $d=d'$ and $\delta[v/d](d')=\delta(d')$ otherwise.
We assume appropriate \emph{data functions} according to operators, and
use $\interpreted{\mathit{exp}}\delta$ to represent a value obtained by the evaluation of a \emph{data expression} $\mathit{exp}$ under a data environment $\delta$.
A data expression interpreted to a value in $\mathbb{B}$ is called a \emph{Boolean expression}.
We write $\vvec{a}$ or $\wvec{a}$ by using boldfaced font or an arrow %$\wvec{\rule{0pt}{0.6ex}\ }$ 
to represent a sequence $a_1,\ldots,a_n$ of objects.
  Especially, $\vvec{d}{:}\vvec{D}$ is an abbreviation of a sequence $d_1{:}D_1, \dots, d_n{:}D_n$.
  We write $\mathbb{D}^*$ as
    a product
    $\mathbb{D}_1 \times\cdots\times \mathbb{D}_n$ of appropriate domains.
In this paper, we use usual operators and constants like $\mathrm{true}$, $\mathrm{false}$, $\le$, $0$, $1$, $+$, $-$, and so on, along with expected data functions.% in examples without stating.

A \emph{Parameterised Boolean Equation System} (PBES) $\Epsilon$ is 
a sequence of well-sorted equations:
  \[ %\Epsilon\ =\
  \left(\sigma_1 X_1(\vvec{d}{:}\vvec{D}) = \varphi_1\right) \mathrel{\cdots}
  \left(\sigma_n X_n(\vvec{d}{:}\vvec{D}) = \varphi_n\right)
  \]
  where $\varphi_i$ is a \emph{predicate formula} defined by the following BNF, 
  and $\sigma_i$ is either one of the quantifiers $\mu,\nu$ used to indicate the least and greatest fixed-points, respectively ($1 \leq i \leq n$). 
  \[
  \begin{array}[t]{rcl}
%   \Epsilon & ::= & \emptyset \ |\
%	\left(\nu X\left(d: D\right) = \varphi\right) \Epsilon\ |\
%	\left(\mu X\left(d:D\right) = \varphi\right) \Epsilon \\
   \varphi & ::= & b \ |\ \varphi \wedge \varphi \ |\ \varphi \vee
	\varphi \ |\ \forall d{:}D\ \varphi \ |\ \exists d{:}D\ \varphi \ |\ X(\vvec{exp})
  \end{array}
  \]
%Here, $\emptyset$ is used for the \emph{empty PBES}, and quantifiers $\mu,\nu$ are used to indicate the least and greatest fixed-points, respectively.
%$\varphi$ is a \emph{predicate formula},
Here $X$ is a predicate variable with fixed arity, % sorted with $\vvec{D} \rightarrow B$,
$b$ is a Boolean expression,  $d$ is a data variable in $\mathcal{V}_D$, 
and $\vvec{exp}$ is a sequence of data expressions. 
We say $\Epsilon$ is \emph{closed} if
it does neither contain free predicate
variables nor free data variables.
Note that the negation is allowed only in expressions
$b$ or $\mathit{exp}$ as a data function.

\begin{example} \label{eg:PBES}
 A PBES $\Epsilon_2$ is given as follows:
 \[
 \begin{array}[t]{rcl}
  \nu X_1(n: N)& = & (n = 0 \wedge X_1(n+2))\ \vee\ \left(n > 0 \wedge X_2\left(n-1\right) \wedge X_1(n+2)\right)\\
  \mu X_2(n: N)& = & (n \geq 3 \wedge X_2(n-2))\ \vee\ \left(n \mathrel{=} 1 \wedge X_1\left(n-1\right)\right)
 \end{array}
 \]
\end{example}
%
%Obviously the occurrences of '$\mathrm{true}$' are redundant in the expressions.
%They are necessary for the subclass introduced in Section~\ref{sec:ePG}.

Since the definition of the semantics is complex,
we omit it and we will explain it by an example.
The formal definition can be found in \cite{GROOTE2005332}.
% before introducing the formal definition.
The meaning of a PBES is determined in the bottom-up order.
Considering a PBES $\Epsilon_2$ in Example~\ref{eg:PBES},
we first look at the second equation, which defines a set $X_2$.
The set $X_2$ is fixed depending on the free variable $X_1$, i.e., the equation should be read as that $X_2$ is the least set satisfying the condition
%\[
``\(
 v \in X_2\ \iff\ (v \geq 3 \wedge v-2 \in X_2)\ \vee\ (v=1 \wedge v-1 \in X_1)
\)''
%\]
for any $v\in \mathbb{N}$.
Thus the set $X_2$ is fixed as
%$\{0,2,4,\dots\}$ 
$\{1,3,5,\dots\}$ if $0\in X_1$; $\emptyset$ otherwise,
%\[  X_2 = 
%   \begin{cases}
%        \{0\}      & \text{if $0\in X_1$}\\
%         \emptyset & \text{otherwise}
%   \end{cases},
%\]
i.e., ``$X_2(v) \iff \mathrm{odd}(v) \wedge X_1(0)$'' for any $v\in \mathbbold{N}$.
Next, we replace the occurrence of $X_2$ in the first equation of $\Epsilon_2$ with $\mathrm{odd}(v) \wedge X_1(0)$, which results in
``\(
   \nu X_1(n: N)\ =\ (n=0 \wedge X_1(n+2)) \ \vee\ (n>0 \wedge \mathrm{odd}(n-1) \wedge X_1(0) \wedge X_1(n+2))
\)'', if we simplify it.
%Since this is simplified as $\nu X_1(n: N)\ =\  X_1(n+1)$
The set
$X_1$ is fixed as the greatest set satisfying that $v \in X_1 \iff (v=0 \wedge v+2 \in  X_1) \ \vee\ (v>0 \wedge \mathrm{odd}(v-1) \wedge 0 \in X_1 \wedge v+2 \in X_1)$ for any $v\in \mathbb{N}$.
All in all, we obtain $X_1=\{0,2,4,\dots\}$ and $X_2=\{1,3,5,\dots\}$.
The solution $\interpreted{\Epsilon}$ of a closed PBES $\Epsilon$
is a
  function which takes a predicate variable, and returns
  a function on $\mathbb{D}^* \to \mathbb{B}$
  that represents
  the corresponding predicate determined by the PBES.
For instance, in the example PBES $\Epsilon_2$,
  $\interpreted{\Epsilon_2}(X_1)$ is the function on $\mathbb{N} \to \mathbb{B}$
  that returns $\mathbbold{t}$ if and only if an even number is given,
  and
  $\interpreted{\Epsilon_2}(X_2)$ is the function on $\mathbb{N} \to \mathbb{B}$
  that returns $\mathbbold{t}$ if and only if an odd number is given.

The \emph{membership problem} for PBESs $\Epsilon$ is a problem that answers whether $X(\vvec{v})$ holds
(more formally $\interpreted{\Epsilon}(X)(\vvec{v})=\mathbbold{t}$) or not for a given predicate variable $X$ and a value $\vvec{v} \in \mathbb{D}^*$.
%In the sequel,
%we intuitively explain proof graphs introduced in~\cite{Cranen2013} in order to characterize the membership problem.
The membership problem is characterized by proof graphs introduced in~\cite{Cranen2013}.
For a PBES 
  \( \Epsilon\ =\
  \left(\sigma_1 X_1\left(\vvec{d{:}\vvec{D}}\right) = \varphi_1\right) \mathrel{\cdots}
  \left(\sigma_n X_n\left(\vvec{d{:}\vvec{D}}\right) = \varphi_n\right),
  \)
the \emph{rank} of $X_i$ ($1\leq i\leq n$) is the number of alternations of $\mu$ and $\nu$
in the sequence $\nu\sigma_1\cdots\sigma_n$.
Note that the rank of $X_i$ bound with $\nu$ is even
and the rank of $X_i$ bound with $\mu$ is odd.
For Example~\ref{eg:PBES}, $\rank_{\Epsilon_2}(X_1)=0$ and $\rank_{\Epsilon_2}(X_2)=1$.
\emph{Bound variables} are predicate variables $X_i$
that occur in the left-hand sides of equations in $\Epsilon$.
The set of bound variables is denoted by $\mathrm{bnd}(\Epsilon)$.
The \emph{signature} $\sig(\Epsilon)$ in $\Epsilon$ is defined by
$\sig(\Epsilon) = \left\{(X_i,\vvec{v}) \mid X_i \in \mathrm{bnd}(\Epsilon),\ \vvec{v} \in \mathbb{D}^* \right\}$.
We use $X_i(\vvec{v})$ to represent $(X_i,\vvec{v})\in \sig(\Epsilon)$.
%\YN{
We use some graph theory terminology to introduce proof graphs.
In a directed graph $\left<V,{\to}\right>$, the \emph{postset} of a vertex $v \in V$ is
the set $\{v' \in V \mid v \to v'\}$.
\begin{definition}\label{def:PG}
 Let $\Epsilon$ be a PBES,
 $V \mathrel{\subseteq} \sig(\Epsilon)$, ${\rightarrow} \mathrel{\subseteq} {V \times V}$, and $r \in \mathbb{B}$.
 The tuple $\left< V, \rightarrow, r \right>$ is called a \emph{proof graph} for the PBES
 if both of the following conditions hold:
 \begin{enumerate}
  \item\label{def:PG:1} 
	   For any $X_i(\vvec{v}) \in V$,
	   $\varphi_i(\vvec{v})$ is evaluated to $r$ 
	   under the assumption that the signatures in the postset of $X_i(\vvec{v})$ are $r$ and
	   the other signatures are $\lnot r$,
	   where $\varphi_i$ is the predicate formula that defines $X_i$.
  \item\label{def:PG:2}
	   For any infinite sequence $Y_0(\vvec{w_0}) \rightarrow Y_1(\vvec{w_1}) \rightarrow \cdots$ in the graph,
%that begins from $X_i(\vvec{v})$,
	   the minimum rank of $Y^{\infty}$ is even,
	   where $Y^{\infty}$ is the set of $Y_j$ that occurs infinitely often in the sequence. 
 \end{enumerate}
\end{definition}

%We use the notation $u^{\bullet}$ for the post set $\{u' \in V \mid u \rightarrow u'\}$
%of a vertex $u$ in a directed graph $\left<V,\rightarrow\right>$.
%\begin{definition}\label{def:PG}
% Let $\left<\Epsilon, \theta, \delta\right>$ be an interpreted PBES,
% $V \mathrel{\subseteq} \sig(\Epsilon)$, ${\rightarrow} \mathrel{\subseteq} {V \times V}$, and $r \in \mathbb{B}$.
% If both of the following conditions hold for any $X_i(\vvec{v}) \in V$,
% the tuple $\left< V, \rightarrow, r \right>$ is called a \emph{proof graph} for the PBES.
% %
% \begin{enumerate}
%  \item\label{def:PG:1} 
%	   $\interpreted{\varphi _i} (\theta [\lnot r / \sig(\Epsilon)][r / X_i(\vvec{v})^{\bullet}]) (\delta [v / d]) = r$
%  \item\label{def:PG:2}
%	   For any infinite sequence $Z_0(\vvec{x_0}) \rightarrow Z_1(\vvec{x_1}) \rightarrow \cdots$ that begins from $X_i(\vvec{v})$,
%	   the minimum rank of $Z^{\infty}$ is even,
%	   where $Z^{\infty}$ is the set of $Z_i$ that occurs infinitely often in the sequence. 
% \end{enumerate}
%\end{definition}
%The condition (\ref{def:PG:1}) says that $\varphi_i = r$ must hold if we assume that the successors of $X_i(\vvec{v})$ are $r$ and
%the other signatures are $\lnot r$.

We say that a proof graph $\left<V, \rightarrow, r\right>$ \emph{proves} $X_i(\vvec{v}) = r$
if and only if $X_i(\vvec{v}) \in V$.
In the sequel, we
  consider the case that $r = \mathbbold{t}$.
  The case $r=\mathbbold{f}$ will derive dual results.
\begin{example} \label{eg:PG}
 Consider the following graph with $r = \mathbbold{t}$ and $\Epsilon_2$ in Example~\ref{eg:PBES}:
%\MS{
 \vspace{1em}
 \[
 \xymatrix@R=1em{
% X_1(0) \ar@(lu,ru) && X_1(2) \ar[rr]\ar[ld] && X_1(4) \ar[ld]\ar[r] & \dots \\
 X_1(0) \ar[rr] && X_1(2) \ar[rr]\ar[ld] && X_1(4) \ar[ld]\ar[r] & \dots \\
 &X_2(1) \ar[lu]&& X_2(3) \ar[ll] && \dots\ar[ll]
 }
% \xymatrix@R=1em{
% X_1(0) \ar[r] & X_1(1) \ar[r] & X_1(2) 
% \ar[r] & \cdots\\
% X_2(0) \ar[u]
% }
 \]
%}
 This graph is a proof graph, %proving \YN{$X_1(2)=\mathbbold{t}$}, 
 which is justified 
% \MS{
 from the following observations:
 \begin{itemize}
 \item The graph satisfies the condition~(1). For example, for a vertex $X_1(2)$, the predicate formula
	   $\varphi_1(2) = (2 = 0 \wedge X_1(2+2))\ \vee\ \left(2 > 0 \wedge X_2\left(2-1\right) \wedge X_1(2+2)\right)$ is $\mathbbold{t}$ assuming that $X_2(1)$ = $X_1(4) = \mathbbold{t}$.
 \item The graph satisfies the condition~(2).  For example, for an infinite sequence $X_1(0) \rightarrow X_1(2) \rightarrow X_2(1) \rightarrow X_1(0) \rightarrow\cdots$, the minimum rank of $\{X_1,X_2\}$ is $0$.
 \end{itemize}
% }
%
% Therefore, this graph satisfies the condition (\ref{def:PG:1}) in Definition~\ref{def:PG}.
% Moreover, $X_1$ occurs infinitely often in every infinite path in the graph and rank of $X_1$ is even.
% Thus, the condition (\ref{def:PG:2}) is satisfied.
\end{example}

The next theorem states the relation between proof graphs and the membership problem on a PBES.
\begin{theorem}[\cite{Cranen2013}]\label{thm:PBES2PG}
 For a PBES $\Epsilon$ and
 a $X_i(\vvec{v}) \in \sig (\Epsilon)$, 
 the existence of a proof graph $\left<V,\rightarrow,r\right>$ 
 such that $X_i(\vvec{v}) \in V$ coincides with
 $\interpreted{\Epsilon}(X_i)(\vvec{v}) = r$.
\end{theorem}

\section{Extended Dependency Spaces}\label{sec:exdepsp}
  This paper discusses an \emph{existential}
  subclass of PBESs where universal-quantifiers are not allowed
  \footnote{This restriction can be relaxed so that universal-quantifiers emerge in $\varphi_{ik}$ of Proposition~\ref{prop:expbes},
      which does not affect the arguments of this paper.
  }.
This class properly includes disjunctive PBESs~\cite{Koolen2015}.
%
%the restricted class of PBESs, defined as follows:
%\begin{definition}\label{def:expbes}
% \emph{Existential PBESs} are subclass of PBES defined by the following grammar:
% \[
% \begin{array}[t]{rcl}
%  \Epsilon & ::= & \emptyset \ |\
%   \left(\nu X\left(d: D\right) = \varphi\right) \Epsilon\ |\
%   \left(\mu X\left(d:D\right) = \varphi\right) \Epsilon \\
%  \varphi & ::= & b \ |\ \varphi \wedge \varphi \ |\ \varphi \vee
%   \varphi \ |\ \exists d{:}D\ \varphi \ |\ X(e)
% \end{array}
% \]
%
% The difference from the original definition is that the quantifier ${\forall}$ does not occurs in $\varphi$.
%\end{definition}
Existential PBESs can be represented in simpler forms as shown in the next proposition.
\begin{restatable}{proposition}{expbes}\label{prop:expbes}
 For every existential PBES $\Epsilon$, there exists an existential PBES $\Epsilon'$ satisfying
 $\interpreted{\Epsilon}(X) = \interpreted{\Epsilon'}(X)$ for all $X \in \mathrm{bnd}(\Epsilon)$,
 and where $\Epsilon'$ is of the following form:
 \[
 \begin{array}[t]{rcl}
  \sigma_1 X_1(\vvec{d}{:}\vvec{D})& = & {\displaystyle \bigvee_{1 \leq k \leq m_1}}
   \exists \vvec{e}{:}\vvec{D}\ \varphi_{1k}(\vvec{d},\vvec{e}) \wedge X_{a_{1k1}}\bigl(\wvec{f_{1k1}(\vvec{d},\vvec{e})}\bigr) \wedge \dots \wedge X_{a_{1kp_{1k}}}\bigl(\wvec{f_{1kp_{1k}}(\vvec{d},\vvec{e})}\bigr) \\
  & \vdots &\\
  \sigma_n X_n(\vvec{d}{:}\vvec{D})& = & {\displaystyle \bigvee_{1 \leq k \leq m_n}}
   \exists \vvec{e}{:}\vvec{D}\ \varphi_{nk}(\vvec{d},\vvec{e}) \wedge X_{a_{nk1}}\bigl(\wvec{f_{nk1}(\vvec{d},\vvec{e})}\bigr) \wedge \dots \wedge X_{a_{nkp_{nk}}}\bigl(\wvec{f_{nkp_{nk}}(\vvec{d},\vvec{e})}\bigr) \\
 \end{array}
 \]
 where $\sigma_i$ is either $\mu$ or $\nu$,
 $\wvec{f_{ikj}(\vvec{d},\vvec{e})}$ is a sequence of data expressions possibly containing variables $\vvec{d},\vvec{e}$,
 and $\varphi_{ik}(\vvec{d},\vvec{e})$ is a
 Boolean expression
% defined by the grammar
% $\varphi ::= b \ |\ \forall d'{:}D\ \varphi \ |\ \exists d'{:}D\ \varphi$,
 containing no free variables except for $\vvec{d},\vvec{e}$.
\end{restatable}
\noindent
 In contrast, a disjunctive PBES is of the following form:
 \[
 \begin{array}[t]{rcl}
  \sigma_1 X_1(\vvec{d}{:}\vvec{D})& = & {\displaystyle \bigvee_{1 \leq k \leq m_1}}
   \exists \vvec{e}{:}\vvec{D}\ \varphi_{1k}(\vvec{d},\vvec{e}) \wedge X_{a_{1k1}}\bigl(\wvec{f_{1k1}(\vvec{d},\vvec{e})}\bigr)\\
  & \vdots &\\
  \sigma_n X_n(\vvec{d}{:}\vvec{D})& = & {\displaystyle \bigvee_{1 \leq k \leq m_n}}
   \exists \vvec{e}{:}\vvec{D}\ \varphi_{nk}(\vvec{d},\vvec{e}) \wedge X_{a_{nk1}}\bigl(\wvec{f_{nk1}(\vvec{d},\vvec{e})}\bigr)\\
 \end{array}
 \]
 We can easily see that disjunctive PBESs are subclass of existential PBESs.
As a terminology, we use \emph{$k$-th clause for $X_i$} to refer to $\exists \vvec{e}{:}\vvec{D}\ \varphi_{ik}(\vvec{d},\vvec{e}) \wedge X_{a_{ik1}}\bigl(\wvec{f_{ik1}(\vvec{d},\vvec{e})}\bigr) \wedge \dots \wedge X_{a_{ikp_{ik}}}\bigl(\wvec{f_{ikp_{ik}}(\vvec{d},\vvec{e})}\bigr)$.
%For closed PBESs, we abbreviate $\interpreted{\Epsilon}\theta\delta$ as $\interpreted{\Epsilon}$.

Hereafter, we extend the notion of dependency spaces~\cite{Koolen2015}, which is designed for disjunctive PBESs, to those for existential PBESs.
The dependency space for a PBES contains all its minimal proof graphs and hence is valuable to find a proof graph.
The dependency space for a disjunctive PBES is a graph consisting of the vertices labelled with $X(\vvec{v})$ for each data $\vvec{v}\in
\mathbb{D}^\ast$ and the edges $X_i(\vvec{v}) \to X_j(\vvec{w})$ for all dependencies meaning that $X_j(\vvec{w}) \implies X_i(\vvec{v})$.
Here $X_j(\vvec{w}) \implies X_i(\vvec{v})$ means that the predicate formula $\varphi_i(\vvec{v})$ of $X_i$ holds under the assumption that $X_j(\vvec{w})$ holds.
A proof graph, if it exists, is found as its subgraph by seeking an infinite path satisfying a condition (2) of Definition~\ref{def:PG}.
This corresponds to choosing one out-going edge for each vertex.
In this sense, the dependency space consists of $\vee$-vertices.
This framework makes sense because a disjunctive PBES contains exactly one predicate variable in each clause.

On the other hand an existential PBES generally contains more than one predicate variable in each clause $\exists \vvec{e}{:}\vvec{D}\ \varphi(\vvec{d},\vvec{e}) \wedge X_{a_1}(\wvec{f_1(\vvec{d},\vvec{e})})\wedge\dots\wedge X_{a_p}(\wvec{f_p(\vvec{d},\vvec{e})})$
  defining $X_i$, which induces dependencies $X_{a_1}(\vvec{w_1})\wedge\dots\wedge X_{a_p}(\vvec{w_p}) \implies X_i(\vvec{v})$
for any data $\vvec{v},\vvec{u} \in \mathbb{D}^*$ such that $\varphi(\vvec{v},\vvec{u})$ and
$\vvec{w_1} = \wvec{f_1(\vvec{v},\vvec{u})}, \dots, \vvec{w_p} = \wvec{f_p(\vvec{v},\vvec{u})}$.
%under the condition 
Hence $\wedge$-vertices are necessary.
Therefore, we extend the notion of dependency spaces by introducing $\wedge$-vertices.
Such dependencies vary according to the clauses.
Thus, 
we need additional parameters $i,k$ 
for $\wedge$-vertices in keeping track of the
$k$-th clause of $X_i$.
For these reasons, 
each $\wedge$-vertex is designed to be a
quadruple $(i,k,\vvec{v},\vvec{u})$.

We illustrate the idea by an example.
Consider an existential PBES $\Epsilon_3$:
\[
 \begin{array}[t]{rcl}
  \nu X_1(n: N)& = & \exists n'{:}N\ \mathrm{even}(n) \wedge X_1(3n+5n') \wedge X_1(4n+5n')
 \end{array}
\]
The dependencies induced from the equation are $X_1(3n+5n')\wedge X_1(4n+5n')\implies X_1(n)$
for each $n'\in \mathbb{N}$ and even $n\in \mathbb{N}$.
Observing the case $n=2$, the dependency $X_1(6+5n')\wedge X_1(8+5n')\implies X_1(2)$ exists for each $n'\in \mathbb{N}$.
In order to show that $X_1(2)$ holds, it is enough that 
we choose one of these dependencies for constructing a proof graph.
Suppose that we will show $X_1(2)$, we must find some $n'$ such that 
%$\mathrm{even}(d),
both $X_1(6+5n')$ and $X_1(8+5n')$ hold.  Thus it is natural to introduce a $\vee$-vertex $X_1(2)$ having edges to $\wedge$-vertices corresponding to $n'$ values.
Each $\wedge$-vertex has out-going edges to $X_1(6+5n')$ and $X_1(8+5n')$.
This is represented in Figure~\ref{fig:depE3},
where $\vee$-vertices are oval and newly-introduced $\wedge$-vertices are rectangular.
\begin{figure}[ht]
 \[\small
 \xymatrix@R=1em{
 && *+[F-:<10pt>]{X_1(2)} \ar[ld]\ar[d]\ar[rd]\\
 & *+[F-]{(1,1,2,0)} \ar[ld]\ar[d] & *+[F-]{(1,1,2,1)} \ar[d]\ar[rd]&\dots\\
 *+[F-:<10pt>]{X_1(6)} & *+[F-:<10pt>]{X_1(8)}  & *+[F-:<10pt>]{X_1(11)} & *+[F-:<10pt>]{X_1(13)}\\
 \vdots & \vdots & \vdots & \vdots
 }
 \]
 \caption{The dependency space of $\Epsilon_3$}
 \label{fig:depE3}
\end{figure}
%Each oval vertex contains $\sig(\Epsilon)$ and
Each $\wedge$-vertex is labelled with $(i,k,\vvec{v},\vvec{w})$ where $i$ and $k$ come from $k$-th clause for $X_i$.
%where $v$ denotes the data where comes from and $w$ denotes the data which we choose for a parameter $e$.
%Intuitively an $\vee$-vertex $X(2)$ holds if some vertex $(\vvec{v},\vvec{w})$ in the succesors $\{(2,0), (2,1), \dots\}$ holds, and
%an $\wedge$-vertex $(2,0)$ holds if the both successors $X(6), X(8)$ hold.

Generally the extended graph consists of $\vee$-vertices $X_i(\vvec{v})$ for all $1\leq i\leq n$ and $\vvec{v}\in \mathbb{D}^*$ and $\wedge$-vertices $(i,k,\vvec{v},\vvec{w})$ for all $1\leq i\leq n$, $k$, $\vvec{v}\in \mathbb{D}^*$, and $\vvec{w}\in \mathbb{D}^*$.
Each $k$-th clause 
$\exists \vvec{e}{:}\vvec{D}\ \varphi(\vvec{d},\vvec{e}) \wedge X_{a_{1}}(\wvec{f_{1}(\vvec{d},\vvec{e})}) \wedge \dots \wedge X_{a_{p}}(\wvec{f_{p}(\vvec{d},\vvec{e})})$
for $X_i$ constructs edges:
\[ X_i(\vvec{v}) \to (i,k,\vvec{v},\vvec{w}),\ \ 
   (i,k,\vvec{v},\vvec{w}) \to X_{a_{1}}(\wvec{f_{1}(\vvec{v},\vvec{w})}),\ \  \dots,\ \ (i,k,\vvec{v},\vvec{w}) \to X_{a_{p}}(\wvec{f_{p}(\vvec{v},\vvec{w})})
\]
for every $\vvec{v}\in \mathbb{D}^\ast$ and $\vvec{w}\in \mathbb{D}^\ast$ such that $\varphi(\vvec{v},\vvec{w})$ holds (see the figure below).
\[\small
 \xymatrix@R=1em{
 & *+[F-:<10pt>]{X_i(\vvec{v})} \ar[d]&\\
 & *+[F-]{(i,k,\vvec{v},\vvec{w})} \ar[ld] \ar@{}[d]|(.6){\dots} \ar[rd] &\\
 *+[F-:<10pt>]{X_{a_1}(\wvec{f_{1}(\vvec{v},\vvec{w})})} & \dots\dots & *+[F-:<10pt>]{X_{a_{p}}(\wvec{f_{p}(\vvec{v},\vvec{w})})}
 }
\]
From now on, we write dependency spaces to refer to the extended one by abbreviating ``extended''.

%
%In the previous example, we take the following sub graph:
%\[
% \xymatrix@R=1em{
%  & *+[F-:<10pt>]{X(2)}\ar[d]\\
%  & *+[F-]{0} \ar[ld]\ar[d]\\
%  *+[F-:<10pt>]{X(6)} & *+[F-:<10pt>]{X(8)}\\
% }
%\]
%Of course, we take edges for each vertices $X(6)$ and $X(8)$ continuously.
%
%
We formalize dependency spaces.
The dependency space for a given PBES $\Epsilon$ is a labeled directed graph $G = (V, E, \Pi)$ such that
\begin{itemize}
 \item $V = \sig(\Epsilon) \cup {\mathbb{N} \times \mathbb{N} \times \mathbb{D}^* \times \mathbb{D}^*}$ is a set of vertices,
 \item $E$ is a set of edges with $E \subseteq V \times V$, which is determined from the above discussion, and
 \item $\Pi : V \to \{\vee, \wedge\}$ is a function which assigns $\vee$/$\wedge$ to all vertices.
\end{itemize}

For the dependency of the PBES $\Epsilon_3$, the next graph is a proof graph of $X_1(2)$.
We get this proof graph by choosing $n'=0$ for every vertex.
 \[
 \xymatrix@R=1em{
 X_1(2) \ar[r]\ar[rd] & X_1(6) \ar[r]\ar[rd] & X_1(18) \ar[r]\ar[rd] & \dots \\
                      & X_1(8) \ar[r]\ar[rd] & X_1(24) \ar[r]\ar[rd] & \dots\\
					  &                      & X_1(32) \ar[r]        & \dots\\
 }
 \]
 We can see that this proof graph is obtained by removing some vertices and collapsing $\wedge$-vertices into the vertex $X_1(v)$
 from the dependency space of $\Epsilon_3$ (Figure~\ref{fig:depE3}).

We show that this property holds in general.
\begin{lemma}\label{lem:depsp2pg}
 For a given existential PBES, if there exist proof graphs of $X(\vvec{v})$, then one of them is obtained from its dependency space
 by removing some $\vee$/$\wedge$-vertices and collapsing $\wedge$-vertices $(i,k,\vvec{v},\vvec{w})$ into the vertex $X_i(\vvec{v})$.
\end{lemma}
Thus in order to obtain a proof graph from the dependency space, we encounter the 
problem that chooses one out-going edge for each $\vee$-node so that the condition (2) of Definition~\ref{def:PG} is satisfied.
This problem corresponds to a problem known as parity games
(see Lemma~\ref{lem:sec3} in the appendix for details,
and parity games with finite nodes are decidable in NP. %(NP).
Moreover, there is a solver, named PGSolver~\cite{parityGame}, which efficiently solves many practical problems.
%
%(\MS{sakai: note that its open that parity game has a PTime algorithm or not})

%In the rest of this paper, we use terms for parity games to \YN{re}present dependency spaces.
Unfortunately,
since dependency spaces have infinite vertices, it is difficult to apply parity game solvers.
Thus we need a way to reduce a dependency space to a finite one as shown in the next section.

\section{Reduced Dependency Space}\label{sec:reduced-depsp}

%In this section, we extend the notion of dependency space for reduced dependency spaces.
In this section, we extend reduced dependency spaces~\cite{Nagae2016} to those for existential PBESs.
We assume that an existential PBES $\Epsilon$ has the form of Proposition~\ref{prop:expbes}.

Given a PBES $\Epsilon$, we define functions $F_{ik}:\sig(\Epsilon) \to 2^B$ and $G_{ik}: B \to 2^{\sig(\Epsilon)}$ for each $k$-th clause for $X_i$ as follows,
where $B$ refers to $\mathbb{N} \times \mathbb{N} \times \mathbb{D}^* \times \mathbb{D}^*$:
\[\begin{array}{r@{}l}
F_{ik}(X_j(\vvec{v})) = &
 \begin{cases}
  \{(i,k,\vvec{v},\vvec{w}) \mid  \varphi_{ik}(\vvec{v},\vvec{w})\} & \text{if $i=j$}\\
  \emptyset & \text{otherwise}
 \end{cases}\\

 G_{ik}(j,k',\vvec{v},\vvec{w}) = &
 \begin{cases}
  \{X_{a_{ik1}}(\wvec{f_{ik1}(\vvec{v},\vvec{w})}), \dots, X_{a_{ikp_{ik}}}(\wvec{f_{ikp_{ik}}(\vvec{v},\vvec{w})})\} & \text{if $i = j \wedge k = k'$}\\
  \emptyset & \text{otherwise}
 \end{cases}
\end{array}\]
%\YN{
%Considering the topic of the edges in dependency spaces,
%}
%\MS{(sakai: I don't understand the meaning of ``the topic of the edges''.)}
Intuitively, $F_{ik}$ is a function that takes a $\vee$-vertex $X_j(\vvec{v})$ and returns $\wedge$-vertices as the successors.
On the other hand,
$G_{ik}$ is a function that takes a $\wedge$-vertex $(j,k',\vvec{v},\vvec{w})$ and returns $\vee$-vertices as the successors.
In other words, $F_{ik}$ and $G_{ik}$ indicate the dependencies.

A reduced dependency space is a graph divided by the congruence relation on the algebra that contains operators $F_{ik}, G_{ik}$.
We formalize this relation.

\begin{definition}
 Let ${\sim}_D, {\sim}_B$ be an equivalence relation on $\sig(\Epsilon)$ and $B$ respectively.
 The pair of relations $\<{\sim}_D, {\sim}_B>$ is \emph{feasible} if all these conditions hold:
 \begin{itemize}
  \item For all $i,j \in \mathbb{N}$, if $i \not= j$ then $X_i(\vvec{v}) \not\sim_D X_j(\vvec{v'})$ for any $\vvec{v}, \vvec{v'} \in \mathbb{D}^*$.
  \item For all $i \in \mathbb{N}$ and $\vvec{v},\vvec{v'} \in \mathbb{D}^*$, if $X_i(\vvec{v}) \sim_D X_i(\vvec{v'})$ then
		$F_{ik}(X_i(\vvec{v})) \sim_B F_{ik}(X_i(\vvec{v'}))$ for any $k$.
  \item For all $i,j \in \mathbb{N}$, if $i \not= j$ or $k \not= k'$ then
		$(i,k,\vvec{v},\vvec{w}) \not\sim_B (j,k',\vvec{v'},\vvec{w'})$ for any $\vvec{v},\vvec{v'},\vvec{w},\vvec{w'} \in \mathbb{D}^*$.
  \item For all $(i,k,\vvec{v},\vvec{w}), (i,k,\vvec{v'},\vvec{w'}) \in B$,
		if $(i,k,\vvec{v},\vvec{w}) \sim_B (i,k,\vvec{v'},\vvec{w'})$ then
		$G_{ik}(i,k,\vvec{v},\vvec{w}) \sim_D G_{ik}(i,k,\vvec{v'},\vvec{w'})$.
 \end{itemize}
 Here, we extend the notion of an equivalence relation ${\sim}$ on some set $A$
 for the equivalence relation on $2^A$ in this way:
 \[
  \alpha, \beta \subseteq A.\ \alpha \sim \beta \iff
 \{[a]_{\sim} \mid a \in \alpha\} = \{[b]_{\sim} \mid b \in \beta\}
 \]
\end{definition}
%We show an useful propery for an extended equivalence relation.
%\begin{proposition}\label{prop:exEq}
% Let $A$ be a set and ${\sim}$ be an equivalence relation on $A$.
% If $\alpha, \beta \subseteq A$ have $\alpha \not\sim \beta$ then, w.l.o.g., there exists $a \in \alpha$ such that $a \not\sim b$ for any $b \in \beta$.
%\end{proposition}

We define a reduced dependency space using a feasible pair of relations and dependency space.
\begin{definition}
 Let $G=(V,E,\Pi)$ be a dependency space of $\Epsilon$.
 For a feasible pair $\<{\sim}_D, {\sim}_B>$ of relations,
 an equivalence relation ${\sim}_G$ on $V$ is defined as ${\sim_D} \cup {\sim_B}$.
 The \emph{reduced dependency space} for a given feasible pair of relations is $G' = (V / {\sim_G}, E', \Pi / {\sim_G})$, where
 \(
 E' = \{([v]_{\sim_G}, [w]_{\sim_G}) \mid (v,w) \in E\}
 \).
\end{definition}
Note that $\Pi/{\sim_G}$ is well-defined from the definition of ${\sim}_G$.

%\YN{
Next theorem states that
the membership problem is reduced to the problem finding a finite reduced dependency space.
\begin{theorem}\label{thm:rds-is-decidable}
% The existence of a proof graph that proves $X(\vvec{v})$ as a subgraph of a dependency space coincides with
% the existence of a proof graph as a subgraph of a reduced dependency space.
 Given a finite reduced dependency space of a PBES,
 then the membership problem of the PBES is decidable.
\end{theorem}
%}

\section{Construction of Reduced Dependency Spaces}\label{sec:gameconst}
In this section, we propose a procedure to construct a feasible pair of relations, i.e., reduced dependency spaces, whose basic idea follows the one in \cite{Nagae2016}.
This seems to be similar to minimization algorithm of automata, but the main difference is on that vertices are infinitely many.
Here we use a logical formula to represent (possibly) infinitely many vertices in a single vertex.
We start from the most degenerated vertices, which corresponds to
a pair $\<{\sim}_D, {\sim}_B>$ of coarse equivalence relations, 
and divide each vertex until the pair becomes feasible.
More specifically, we start from the $\vee$-vertices
$\{X_1(\vvec{v})\mid \vvec{v}\in\mathbb{D}^*\}, \dots, \{X_n(\vvec{v})\mid \vvec{v}\in\mathbb{D}^*\}$
and $\wedge$-vertices
$\{(1,1,\vvec{v},\vvec{w})\mid \vvec{v},\vvec{w}\in\mathbb{D}^*\}, \dots, \{(n,m_n,\vvec{v},\vvec{w}) \mid \vvec{v},\vvec{w}\in\mathbb{D}^*\}$.
The procedure keeps track of a partition of the set $\mathbb{D}^*$
for each $i\in\{1,\ldots,n\}$ and a partition of the set
$\mathbb{D}^* \times \mathbb{D}^*$ for each $i \in \{1,\dots,n\}, k \in \{1,\dots,m_i\}$,
and makes partitions finer.

%At the end of this section, we prove the soundness of the procedure.

%\YN{
Recall that a \emph{partition} of a set $A$ is a family $\Phi$ of sets satisfying
$\bigcup_{\phi \in \Phi}\phi = A$ and
$\forall \phi, \phi' \in \Phi.\ \phi \neq \phi' \implies \phi \cap \phi' = \emptyset$.
For a given PBES $\Epsilon$, we call a family $\mathcal{P}$ of partitions is a
\emph{partition family of $\Epsilon$}
if $\mathcal{P}$ has the form $\<\Phi_1,\dots,\Phi_n,\Psi_{11},\dots,\Psi_{nm_n}>$,
and satisfies the following conditions:
\begin{itemize}
 \item $\Phi_i$ is a partition of $\mathbb{D}^*$ for every $i$, and
 \item $\Psi_{ik}$ is a partition of $\mathbb{D}^* \times \mathbb{D}^*$ for every $i,k$.
\end{itemize}
%}\noindent
Every element of a partition family $\mathcal{P}$ is a partition of $\mathbb{D}^*$ or $\mathbb{D}^* \times \mathbb{D}^*$,
hence we naturally define an equivalence relation 
${\sim}^D_{\mathcal{P}}$ on $\mathbb{D}^*$
and ${\sim}^B_{\mathcal{P}}$ on $B$.

We define a function $H$ that takes a partition family and returns another partition family
obtained by doing necessary division operations to its elements.
The procedure repeatedly applies $H$ to the initial partition family until it
saturates.  If it halts, the resulting tuple induces a reduced dependency space.
In the procedure,
functions $\mathbb{D}^* \to \mathbb{B}$ (resp.\ $\mathbb{D}^* \times \mathbb{D}^* \to \mathbb{B}$) represented by Boolean expressions with lambda binding
are used to represent an infinite subset of a data domain $\mathbb{D}^*$ (resp.\ $\mathbb{D}^* \times \mathbb{D}^*$),
In other words,
a function $f = \lambda \vvec{d}{:}D^*. \phi$ (resp.\ $g = \lambda (\vvec{d},\vvec{e}){:}D^*\times D^*. \phi$) can be regarded as a set
$\{\vvec{v} \in \mathbb{D}^* \mid f(\vvec{v}) = \mathbbold{t}\}$
(resp. $\{(\vvec{v},\vvec{w}) \in \mathbb{D}^* \times \mathbb{D}^* \mid g(\vvec{v},\vvec{w}) = \mathbbold{t}\}$).
In the sequel, we write Boolean functions for the corresponding sets.
%In the sequel, we abuse operations on sets to denote Boolean operations.
%For example, we may use the binary operators $\cap$ (resp.\ $\subseteq$) on sets for intersection (resp.\ implication) in Boolean expressions.

The division function $H$ consists of two steps, the division of $\Phi$ and $\Psi$.
We give an intuitive explanation of the division $\Phi = \{\mathbb{N}\}$
  by a set $f = \lambda d{:}N.\ \exists e{:}N\ d+e < 10$, where assuming that $d+e<10$ appears in a PBES as $\varphi_{ik}(d,e)$.
Recall that the function $F_{ik}$ is defined by $F_{ik}(X_i(v)) = \{(i,k,v,w) \mid \varphi_{ik}(v,w)\}$
and the parameter $e$ is quantified by $\exists$ in existential PBESs.
We have to divide the data domain $\mathbb{N}$ into its intersection with $f$ and the rest,
i.e., $\mathbb{N}\cap f = \{v \in \mathbb{N} \mid v < 10\}$ and $\mathbb{N}\cap \overline{f} = \{v \in \mathbb{N} \mid v \geq 10\}$, where $\overline{f}$ denotes the complement of a set $f$.
This division is necessary
because the feasibility condition requires that the mapped values of $F_{ik}$ are also 
% equivalent.
in the same set, and hence
we must separate $v \in \mathbb{N}$
according to
whether $F_{ik}(X_i(v))$ is empty or not.

Next, suppose $\Phi$ is divided into two blocks $\{v \mid v < 10\}$ and $\{v \mid v\geq 10\}$ in the first step.
We assume that a formula $X(d+e)$
appears in the predicate formula of a PBES.
Then, we have to divide $\Psi = \{\mathbb{N} \times \mathbb{N}\}$ into $\{(v,w) \mid v+w < 10\}$ and $\{(v,w) \mid v+w \geq 10\}$.
This is because the feasibility condition requires that the mapped values of $G_{ik}$ are also
%equivalent.
in the same set.

Now we formalize these operations.
We must divide each set $\phi$ in a partition $\Phi$ according to a set $\psi$ in a partition $\Psi$, and similarly
divide each set $\psi$ in a partition $\Psi$ according to 
 a set $\phi$ in a partition $\Phi$.
For the definition, we prepare some kind of the inverse operation for $F_{ik}$ and $G_{ik}$.
%\MS{
\[
 \begin{array}{lll}
  F'_{ik}(\psi) &=& \{ \vvec{v}
             \mid (i,k,\vvec{v},\vvec{w})\in F_{ik}(X_i(\vvec{v})),\,
                  (\vvec{v},\vvec{w})\in \psi \}\\
  G'_{ik}(\phi) &=& \{ (\vvec{v},\vvec{w})
             \mid G_{ik}(i,k,\vvec{v},\vvec{w}) 
                  \cap \phi \neq \emptyset \}
 \end{array}
\]
Then the division operations are given as follows:
\[
 \begin{array}{lll}
  \Phi \otimes_{ik}^D \psi &=& 
    \{ F'_{ik}(\psi) \cap \phi \mid \phi \in \Phi\}
    \cup
    \{ \overline{F'_{ik}(\psi)} \cap \phi 
        \mid \phi \in \Phi\}\\
  \Psi \otimes_{ik}^B \phi &=& 
    \{ G'_{ik}(\phi) \cap \psi \mid \psi \in \Psi\}
    \cup
    \{ \overline{G'_{ik}(\phi)} \cap \psi 
        \mid \psi \in \Psi\}\\
 \end{array}
\]
%}
%
The operator ${\otimes_{ik}^D}$ obviously satisfies $(\Phi \otimes_{ik}^D \psi_1) \otimes_{ik}^D \psi_2 = (\Phi \otimes_{ik}^D \psi_2) \otimes_{ik}^D \psi_1$, thus
we can naturally extend it on sets of formulas as follows:
\[
 \Phi \otimes_{ik}^D \{\psi_1,\dots,\psi_p\} = \Phi \otimes_{ik}^D \psi_1 \otimes_{ik}^D \cdots \otimes_{ik}^D \psi_p
\]
Also, it is easily shown that if $\Phi$ is a partition of $\mathbb{D}^*$, then $\Phi \otimes_{ik}^D \Psi$ is also a partition for a set $\Psi'$ of formulas.
These facts are the same in the case of operator ${\otimes_{ik}^B}$.

We unify these operators in a function that refines a given partition family.
%gets a partition family and returns a refined partition family.
\begin{definition}\label{def:proc_step}
 Let $\mathcal{P}$ be $\<\Phi_1,\ldots,\Phi_n,\Psi_{11},\ldots,\Psi_{nm_n}>$.
 The partition functions $H_{ik}^D, H_{ik}^B$ for each $i$ and $k$ are defined as follows:
 \[
 \begin{array}{rcl}
  H_{ik}^D(\mathcal{P}) &=& \<\ldots,\Phi_{i-1},\Phi_i \otimes_{ik}^D \Psi_{ik}, \Phi_{i+1}, \dots>\\
  H_{ik}^B(\mathcal{P}) &=& \<\ldots,\Psi_{i(k-1)},\Psi_{ik} \otimes_{ik}^B \Delta_D, \Psi_{i(k+1)}, \dots>\\
%  \YN{\Delta_B} &=& \YN{\{\{(i,k,\vvec{v},\vvec{w}) \mid (\vvec{v},\vvec{w}) \in \psi\} \mid \psi \in \Psi_{ik}\}}\\
  \Delta_D &=& %\YN{
   \bigcup_{1\leq j \leq p_{ik}}
   \{\{X_{a_{ikj}}(\vvec{v}) \mid \vvec{v} \in \phi\}
   \mid \phi \in \Phi_{a_{ikj}}\}
%   }
 \end{array}
 \]
 We bundle these functions as $H^D = H_{nm_n}^D \circ \dots \circ H_{11}^D, H^B = H_{nm_n}^B \circ \dots \circ H_{11}^B$ and $H = H^B \circ H^D$
 by composition $\circ$.
\end{definition}

We define the partition procedure that applies the partition function $H$ to the trivial partition family
$\mathcal{P}_0 = \<\{\mathbb{D}^*\},\dots,\{\mathbb{D}^*\},\{\mathbb{D}^* \times \mathbb{D}^*\},\dots,\{\mathbb{D}^*\times \mathbb{D}^*\}>$ until it saturates.
We write the family of the partitions obtained from the procedure as $H^{\infty}(\mathcal{P}_0)$.

\begin{example}\label{ex:partition}
 Consider the existential PBES $\Epsilon_4$ given as follows:
  \[
  \begin{array}[t]{rcl}
   \nu X_1(n{:}N)& = & \exists n'{:}N\ \mathrm{even}(n) \wedge X_1(3n+5n') \wedge X_1(4n+5n')
  \end{array}
  \]
The reduced dependency space for $\Epsilon_4$ is:
\[\small
 \xymatrix@R=1em{
 *+[F-:<10pt>]{X_1(N_0)} \ar[dd] \ar[rdd]& *+[F-:<10pt>]{X_1(N_1)}\\
 \\
 *+[F-:]{B_{00}} \ar@<1ex>[uu]& *+[F-:]{B_{01}} \ar[uu] & *+[F-:]{B_{10}}\ar[lluu]\ar[luu] & *+[F-:]{B_{11}}\ar[llluu]\ar[lluu]\\
 }
 \]
where $N_0 = \{0,2,\dots\}, N_1 = \{1,3,\dots\}$ and $B_{ik} = \{(1,1,d,e) \mid d \equiv_2 i \wedge e \equiv_2 k\}$ for each $i \in \{0,1\}$ and $k \in \{0,1\}$.

 In order to construct this, we apply $H$ to the initial partition $\<\{\lambda n{:}N.true\},\{\lambda(n,n'){:}N^2.true\}>$.
 First, we apply $H^D$:
 \[
 \begin{array}{rcl}
  H^D(\<\{\lambda n{:}N.true\},\{\lambda(n,n'){:}N^2.true\}>) &=& (\<\{\lambda n{:}N.true\} \otimes_{11}^D \{\lambda n{:}N.\mathrm{even}(n)\},\{\lambda(n,n'){:}N^2.true\}>)\\
  &=& \<\{\lambda n{:}N. \mathrm{even}(n), \lambda n{:}N. \lnot\mathrm{even}(n)\},\{\lambda(n,n'){:}N^2.true\}>
 \end{array}
 \]
 We write $\{\lambda n{:}N. \mathrm{even}(n), \lambda n{:}N. \lnot\mathrm{even}(n)\}$ as $\Phi$ for readability.
 Next, we apply $H^B$:
 \[
 \begin{array}{r@{~}l@{~}l}
  H^B(\langle\Phi,\{\lambda(n,n'){:}N^2.true\}\rangle) =& \<\Phi,&\{\lambda(n,n'){:}N^2.true\} \otimes_{11}^B \Phi>\\
  =& \langle\Phi,&
   \{
   \lambda (n,n'){:}N^2. \mathrm{even}(n)\wedge\mathrm{even}(n'),
   \lambda (n,n'){:}N^2. \mathrm{even}(n)\wedge\lnot\mathrm{even}(n'),\\
  && \lambda (n,n'){:}N^2. \lnot\mathrm{even}(n)\wedge\mathrm{even}(n'),
   \lambda (n,n'){:}N^2. \lnot\mathrm{even}(n)\wedge\lnot\mathrm{even}(n')
   \}\rangle
 \end{array}
 \]
 The resulting partition family is a fixed-point of $H$, and hence
   the procedure stops.
 This partition family induces the set of vertices in the reduced dependency space.
\end{example}

We show that the procedure returns a partition family which induces a feasible pair of relations, i.e.,  reduced dependency space.
\begin{restatable}{theorem}{validityproc}\label{thm:proc}
 Suppose the procedure terminates and returns a partitions family $\mathcal{P} = H^{\infty}(\mathcal{P}_0)$.
 Then, the pair ${\sim}^D_{\mathcal{P}}$ and ${\sim}^B_{\mathcal{P}}$ is feasible.
\end{restatable}

\section{Implementation and an Example: Downsized McCarthy 91 Function}
This section states implementation issues of
the procedure presented in Section~\ref{sec:gameconst} and
a bit more complex example, which is inspired by the McCarthy 91 function.

We describe the overview of our implementation.
For a given PBES, the first step calculates a partitions family $\mathcal{P} = H^\infty(\mathcal{P}_0)$
  by repeatedly applying the function $H$ in Definition~\ref{def:proc_step}.
  This step requires a lot of SMT-solver calls.
  We'll explain the implementation in the next paragraph.
  Once the procedure terminates, the obtained partitions family $\mathcal{P}$ determines a feasible pair ${\sim}^D_{\mathcal{P}}$ and ${\sim}^B_{\mathcal{P}}$ by Theorem~\ref{thm:proc}.
  Considering the finite dependency space constructed from the pair as a parity game, the second step constructs a proof graph by using PGSolver, which is justified by the proof of Theorem~\ref{thm:rds-is-decidable}.
(See Lemma~\ref{lem:sec3} and Lemma~\ref{lem:dseqrds} stating that
the existence of a proof graph can be checked by solving
a parity game on a reduced dependency space.)
%Therefore, we regard the reduced dependency space obtained from the first step
%as a parity game, and solve it using PGSolver.

The key to the implementation of $H^\infty(\mathcal{P}_0)$
is the operators ${\otimes}_{ik}^D$ and ${\otimes}_{ik}^B$ used in the function $H$.
We focus on this and discuss how to implement these operators.
We use Boolean expressions with lambda binding to represent subsets of data domains $\mathbb{D}^\ast$ and $\mathbb{D}^\ast \times \mathbb{D}^\ast$ as used for the intuitive explanation of the procedure and Example~\ref{ex:partition}.
Then it seems as if it would be simple to implement the procedure.
It, however, induces non-termination without help of SMT solvers.
Let us look more closely at the division operation $\Phi \otimes_{ik}^D \psi$.
Suppose that $\lambda (\vvec{d},\vvec{e}){:}D^*\times D^*. \hat\psi$ is given as an argument of $F'_{ik}$.  The resulting function is presented as
$\lambda \vvec{d}{:}D^*.\hat\eta$ where
\(
  \hat\eta = \exists \vvec{e}{:}D^* (\varphi_{ik}(\vvec{d},\vvec{e}) \wedge \hat\psi).
\)
By using a set of functions to represent a partition $\Phi$,
each division of $\lambda \vvec{d}{:}D^*.\hat\phi$ in $\Phi$ by 
$F'_{ik}(\psi)$ is simply implemented;
it produces two functions $\lambda \vvec{d}{:}D^*.(\hat\eta \wedge \hat\phi)$ and $\lambda \vvec{d}{:}D^*.(\neg\hat\eta \wedge \hat\phi)$.
This simple symbolic treatment always causes non-termination of the procedure without removing an empty set from the partition.  Since the set represented by a function $\lambda \vvec{d}{:}D^*.\hat\phi$ is empty if and only if $\hat\phi$ is unsatisfiable, this can be done by using an SMT solver.
An incomplete unsatisfiability check easily causes a non-termination of the procedure, even if the procedure with complete unsatisfiability check terminates.
Thus, the unsatisfiability check of Boolean expressions is one of the most important issues in implementing the procedure.
For instance, the examples illustrated in this paper are all in the class of Presburger arithmetic, which is the first-order theory of the natural numbers with addition.
It is known that the unsatisfiability check of Boolean expressions in this class is decidable~\cite{presburger1931vollstandigkeit}.

Consider the following function $F$ on $\mathbb{N}$ determined by a given $a\in\mathbb{N}$:
\[
 F(n) = \begin{cases}
		 n - 1 & \text{if $n > a$}\\
		 F(F(n+2)) & \text{if $n \leq a$}\\
		\end{cases}
\]
The function $F(n)$ returns $n-1$ if $n > a$, and returns $a$ otherwise.
The latter property $F(n) = a$ for $n \leq a$
can be proved by induction on $n - k$ where $k \in \mathbb{N}$.

For an instance $a = 3$, this
function can be modeled by the following existential PBES:
\[
\begin{array}{rcl}
 \mu M(x{:}N,y{:}N) &=& (x > 3 \wedge y+1 = x \wedge X_T)
  \vee (\exists e{:}N\ x \leq 3 \wedge M(x+2,e) \wedge M(e,y))\\
 \nu X_T &=& X_T
\end{array}
\]
where $X_T$ is a trivial predicate variable which denotes true.

To understand this modeling, we consider the case $x = 0$.
In this case, because the first clause does not hold for any $y$,
$M(0,y)$ holds only if $M(2,e)$ and $M(e,y)$ hold for some $e$.
This implies $y = F(0) \iff \exists e{:}N\ e = F(0+2) \wedge y = F(e)$.
In addition, in the case $x > 3$, $M(x,y)$ holds if $y+1 = x$,
that is equivalent to $y = F(x)$.
From this consideration, $M(x,y)$ holds if and only if $y = F(x)$.

To solve this example, we have implemented the procedure, which uses SMT solver Z3~\cite{DeMoura:2008:ZES:1792734.1792766} for deciding the emptiness of sets in the division.
We attempted to solve the above PBES,
and got the reduced dependency space consisting of 65 nodes in a few seconds.
A proof graph is immediately found from the resulting space by applying PGSolver~{\cite{parityGame}}.
Figure~\ref{fig:f3-depsp} displays a part of the obtained graph
consisting of vertices where $M(x,y)$ holds.
\begin{figure}[ht]
 \center
 \includegraphics[width=15cm]{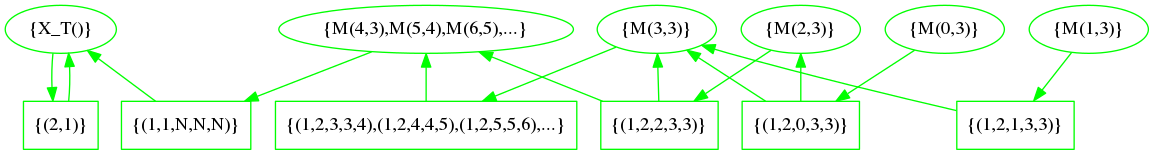}
 \caption{A part of the reduced dependency space where $M(x,y)$ holds.}
 \label{fig:f3-depsp}
\end{figure}
Although
a proof graph induced from the reduced dependency space is finite,
the reduced dependency space is nevertheless useful because the search space is infinite.
We also tried to solve a larger instance $a=10$ and got the spaces consisting of 394 nodes in 332 seconds,
in which Z3 solver spent 325 seconds.

%\MS{
For another example, a disjunctive PBESs for a trading problem~\cite{Nagae2016} is successfully solved by our implementation in a second, which produced
a space consisting 12 nodes.
Note that disjunctive PBESs is a subclass of existential PBESs.
%}

%\YN{
%  \MS{
    In contrast, the procedure does not halt for the PBES $\Epsilon_2$ in Example~\ref{eg:PBES}, nor even for the next simple PBES:
    \[
    \mu X(d{:}N) \,=\, (d = 0) \vee (d > 0 \wedge X(d-1))
    \]
    where its solution is $X = \mathbb{N}$.
%  }
Our procedure starts from the entire set $\mathbb{N}$ and divides it into $\{0\}$ and $\{1,2,\dots\}$ because of the first clause.
After that, the latter set is split into $\{1\}$ and $\{2,3,\dots\}$ using the second clause.
Endlessly, the procedure splits the latter set into the minimum number and the others.
%Obviously, there exists finite proof graph of $X(n)$ for $n \in \mathbb{N}$.
%However, the procedure divides the data domain irrespective of $X(n)$.
%  \MS{
  From this observation, the feasibility condition on $\<\sim_D,\sim_B>$ may be too strong, and weaker one is promising.

\section{Conclusion}
We have extended reduced dependency spaces for existential PBESs, and
have shown that a proof graph is obtained from the space by solving parity games
if the space is finite.
Reduced dependency spaces are valuable because the dependency spaces for most of PBESs are infinite, but existential PBESs may have a finite reduced dependency space.
We also have shown a procedure to construct reduced dependency spaces and have shown the correctness.
We have shown some examples including a downsized McCarthy 91 function is successfully characterized by our method by applying an implementation.

Reduced dependency space is defined so that it contains all minimal proof graphs.
For a membership problem to obtain a proof graph which proves $X(\vvec{v})$,
a reduced space may require too many division to make the entire data domain consistent.
This sometimes induces the loss of termination of the procedure.
Proposing a more clever procedure is one of the future works.
Moreover, our implementation relies on the shape of PBESs, not sets defined by them.
This is indicated by the example in the last of section~6.
To clarify these conditions is also our future works.

%%%% ACKNOWLEDGEMENT
%\subsection*{Acknowledgements}
%
%This paper is partially supported by JSPS KAKENHI Grant Number JP17H01721.
%%We thank the anonymous reviewers very much for their useful comments to improve this paper.  

\nocite{*}
\bibliographystyle{eptcs}

\newpage
\appendix
\section{Proof of Proposition \ref{prop:expbes}}
\begin{definition}\label{def:expbes}
\emph{Existential PBESs} are subclass of PBES defined by the following grammar:
\[
\begin{array}[t]{rcl}
 \Epsilon & ::= & \emptyset \ |\
  \left(\nu X\left(d: D\right) = \varphi\right) \Epsilon\ |\
  \left(\mu X\left(d:D\right) = \varphi\right) \Epsilon \\
 \varphi & ::= & b \ |\ \varphi \wedge \varphi \ |\ \varphi \vee
  \varphi \ |\ \exists d{:}D\ \varphi \ |\ X(\vvec{exp})
\end{array}
\]

The difference from the original definition is that the universal-quantifier does not occur in $\varphi$.
\end{definition}

\expbes*
\begin{proof}
 We can normalize each $\varphi$ by the following steps.
\begin{enumerate}
 \item Rename all bound variables so that different scope variables are distinct.
 \item Lift up all existential quantifiers to obtain the form $\exists \vvec{e}{:}D.\ \varphi'$.
 \item Calculate the disjunctive normal form of $\varphi'$ and obtain $\bigvee_R \varphi_R$.
	   Then $\bigvee_R \exists \vvec{e}{:}D.\ \varphi_R$ is an expected form.
\end{enumerate}
\QED
\end{proof}

\section{Proofs related to proof graphs}
\subsection{Parity games and dependency spaces}
\begin{definition}
 A \emph{parity game} is a directed graph $G = (V, {\to}, \Omega, \Pi)$ such that
 \begin{itemize}
  \item $V$ is a set of vertices,
  \item ${\to}$ is a set of edges with ${\to} \subseteq V \times V$,
  \item $\Omega : V \to \mathbb{N}$ is a function which assigns \emph{priority} to each vertex, and
  \item $\Pi : V \to \{\circ, \Box\}$ is a function which assigns \emph{player} to each vertex.
 \end{itemize}
 Parity game is a game played by two players $\circ$ and $\Box$.
 This game progresses by moving the piece placed at a vertex along an edge.
 The player $\Pi(v)$ moves the piece when the piece is on a vertex $v$.
 We call a parity game started from $v$ if the piece is initially placed on the vertex $v$.
 A play of a parity game is a sequence of vertices that the piece goes through.
 The player $\circ$ \emph{wins} the game if and only if the player $\Box$ cannot move the piece,
 or the largest priority that occurs infinitely often in the play is an even number when the play continues infinitely.
 We write $\rank^\infty(\pi)$ as this largest priority for the play $\pi$.
 \end{definition}
 For a given parity game $G$ started from a fixed vertex, a strategy $S$ of a player is a function
 that selects a vertex to move along an edge, given a vertex owned by the player.
 A strategy $S$ is called a \emph{winning strategy} when
 the player wins the game $G$ according to $S$ even if the opponent player plays optimally.

 The next proposition shows the determinacy of the winner for parity games.
 \begin{proposition}[\cite{emerson1991tree}]\label{prop:pg-determinacy}
  Given a parity game, starting vertex dominates the winner, which
  means that either of the player $\circ$ or $\Box$ has a winning strategy
  according to the vertex started from.
 \end{proposition}
 We say that \emph{a player wins the game on a vertex} if he has a winning strategy started from the vertex.

%We will convert the dependency space for a PBES into a parity game so that
%if there exists a proof graph containing a specified $\vee$-vertex, then $\circ$ wins the game on the vertex.
A dependency space for a PBES $\Epsilon$ is regarded as a parity game.
The function $\Pi$ naturally defined from the shape of $\vee$/$\wedge$-vertices,
  i.e., $\Pi(v)={\circ}$ if and only if $v$ is an $\vee$-vertex.
We give the priority function $\Omega$ as follows:
\[
\Omega(v) = \left\{
		\begin{array}{ll}
		 u - \rank(X_i), & \text{if } v = X_i(\vvec{v}) \in \sig(\Epsilon)\\
		 0, & \text{otherwise}
		\end{array}
		\right.
\]
where $u$ is the minimum even number satisfying $u \geq \rank(X_i)$ for any $X_i \in \mathrm{bnd}(\Epsilon)$.
We call this \emph{a parity game obtained from $\Epsilon$}.
%By the above transformation, we can regard dependency spaces as parity games.
%
Note that $V$ and ${\to}$ in dependency spaces are generally infinite sets, while $\{\Omega(v) \mid v \in V\}$ is finite.
Thus each parity game corresponding to a dependency space is well-defined.

Moreover, a reduced dependency space obtained from a feasible relation ${\sim}$ is also regarded as a parity game,
because $\Omega / {\sim}$ and $\Pi / {\sim}$ are well-defined.
We call this \emph{a reduced parity game $G / {\sim}$} for a parity game $G$ and a feasible relation ${\sim}$.

\subsection{Proof of Lemma~\ref{lem:depsp2pg}}
This lemma is justified by Lemma~\ref{lem:sec3}.
\begin{lemma}\label{lem:sec3}
 For a given PBES $\Epsilon$, there exists a proof graph that proves $X(\vvec{v}_0)$
 if and only if
 the player $\circ$ wins 
 the game $G$, obtained from $\Epsilon$, on $X(\vvec{v}_0)$.
\end{lemma}
\begin{proof}
 We use the notation $X_i(\vvec{v})^\bullet$ to represent the postset of $X_i(\vvec{v})$ for the target graph.

 $\Rightarrow$) Let $P$ be a proof graph of $X(\vvec{v}_0)$.
 By the form of existential PBESs and the condition (\ref{def:PG:1}) of proof graphs,
 for every vertex $X_i(\vvec{v})$ in $P$ there exist $k$ and $\vvec{w}$ such that
 \[
 \varphi_{ik}(\vvec{v},\vvec{w}) \text{ and }
 X_i(\vvec{v})^\bullet \supseteq\\ \{X_{a_{ik1}}(f_{ik1}(\vvec{v},\vvec{w})),\dots,X_{a_{ikp_{ik}}}(f_{ikp_{ik}}(\vvec{v},\vvec{w}))\}.
 \]
 From this, we take the strategy of the player $\circ$:
 when the piece is on a vertex $X_i(\vvec{v}) \in P$, move the piece to $(i,k,\vvec{v},\vvec{w})$ for $k,\vvec{w}$
 determined by the above condition.

 We show that the piece never goes on the $\vee$-vertices not in $P$ with the parity game started from $X(\vvec{v_0})$
 by induction of the steps of the game.
 First, the piece is on $X(\vvec{v}_0)$, and it is in $P$.
 Assume the piece is on $X_i(\vvec{v}) \in P$.
 Let $(i,k,\vvec{v},\vvec{w})$ be the $\wedge$-vertex which the piece goes on under the strategy.
 From the definition of the strategy, all vertices to which the player $\Box$ can move the piece are in $P$.
 Thus, the piece never goes on the $\vee$-vertices not in $P$, and the strategy is well-defined.
 
 We show that the strategy is winning strategy started from $X(\vvec{v}_0)$.
 If the piece cannot move on a vertex $X_i(\vvec{v})$, then $\varphi_{ik}(\vvec{v},\vvec{w})$ never holds for any $k, \vvec{w}$
 by the definition of $G$.
 Then, it contradicts the condition (\ref{def:PG:1}) of proof graphs.
 In addition, suppose that there exists an infinite parity game in which the player $\circ$ loses.
 Let $\pi$ be the infinite path tracing the movement of the piece.
 We define an infinite path $\pi'$ from $\pi$ by collapsing the $\wedge$-vertices into $\vee$-vertices.
 Because $\pi'$ only consists of $\vee$-vertices, $\pi'$ is also a path of $P$.
 $P$ is a proof graph, therefore, the minimum rank of $X^\infty$ with $\pi'$ is even by the condition~(\ref{def:PG:2}) of proof graphs.
 In contrast, the largest priority on $\pi'$ is odd because the player $\circ$ loses.
 The minimum rank with $\pi'$ is odd by the definition of the priority function $\Omega$, however, this contradicts.

 $\Leftarrow$) Assume that the player $\circ$ has a winning strategy.
 We construct a graph $P$:
 \begin{itemize}
  \item $X(\vvec{v}_0)$ is in $P$.
  \item Suppose that $X_i(\vvec{v})$ is in $P$ and $(i,k,\vvec{v},\vvec{w})$ is the vertex to which the player $\circ$ 
		moves the piece on $X_i(\vvec{v})$.
		Then, $(i,k,\vvec{v},\vvec{w})^\bullet$ is in $P$ and there exist edges from $X_i(\vvec{v})$ to each vertex in $(i,k,\vvec{v},\vvec{w})^\bullet$.
 \end{itemize}
 It is obvious that the graph $P$ have the conditions of proof graphs.
 \QED
\end{proof}

\subsection{Proof of Theorem \ref{thm:rds-is-decidable}}
\begin{proposition}\label{prop:congr}
 Let $G$ be a parity game and $G / {\sim}$ be a reduced parity game for a PBES.
 Then, for all vertex $[A]_{\sim}, [B]_{\sim} \in G / {\sim}$ and for all $u \in [A]_{\sim}$,
 $[A]_{\sim} \to [B]_{\sim}$ implies $\exists v \in [B]_{\sim}.\ u \to v$.
\end{proposition}

The theorem~\ref{thm:rds-is-decidable} holds immediate from the following lemma and the fact that solving a finite parity game is decidable.
\begin{lemma}\label{lem:dseqrds}
 The player $\circ$ wins on $X(\vvec{v})$ for a dependency space if and only if the player $\circ$ wins on $[X(\vvec{v})]_{\sim}$ for a reduced dependency space.
\end{lemma}
\begin{proof}
 $\Rightarrow$)  We prove contraposition.
 Suppose the player $\circ$ loses on $[X(\vvec{v})]_{\sim}$.
 This means that the player $\Box$ wins on $[X(\vvec{v})]_{\sim}$ by the proposition~\ref{prop:pg-determinacy}.
 Let $S'$ be a winning strategy of the player $\Box$ on $[X(\vvec{v})]_{\sim}$.
 Then, a strategy $S$ of the player $\Box$ on $X(\vvec{v})$ can be defined as $S(u) = v$ for $u \in G$ where $v \in S'([u]_{\sim})$
 from the proposition~\ref{prop:congr}.

 Any game $P$ on $X(\vvec{v})$ according to the strategy $S$ is corresponding to a game $P'$ on $[X(\vvec{v})]_{\sim}$
 according to the strategy $S'$.
 That is, the sequence of the rank $P$ is equal to the sequence of $P'$.
 $S'$ is a winning strategy, therefore $S$ is also a winning strategy.

 $\Leftarrow$) Suppose the player $\circ$ has a winning strategy $S'$ on $[X(\vvec{v})]_{\sim}$.
 We can define a strategy $S$ on $X(\vvec{v})$ using $S'$ in a similar way, and $S$ is also a winning strategy.
 \QED
\end{proof}

\subsection{Proof of Theorem \ref{thm:proc}}
\validityproc*
\begin{proof}
 Proof by contradiction.
 Let $\mathcal{P} = \<\dots,\Phi_i^\infty,\dots,\Psi_{ik}^\infty,\dots>$ be a fixed-point of $H$
 and ${\sim_\mathcal{P}} = {\sim_\mathcal{P}^D} \cup {\sim_\mathcal{P}^B}$.
 If $\mathcal{P}$ is not feasible, then at least one of the following conditions holds:
 \begin{enumerate}
  \item There exists $X_i(\vvec{v})$ and $X_i(\vvec{v}')$ such that $X_i(\vvec{v}) \sim_{\mathcal{P}} X_i(\vvec{v}')$
		and $F_{ik}(X_i(\vvec{v})) \not\sim_{\mathcal{P}} F_{ik}(X_i(\vvec{v}'))$ for some $k$.
  \item There exists $(i,k,\vvec{v},\vvec{w})$ and $(i,k,\vvec{v'},\vvec{w'})$ such that $(i,k,\vvec{v},\vvec{w}) \sim_{\mathcal{P}} (i,k,\vvec{v'},\vvec{w'})$
		and $G_{ik}(i,k,\vvec{v},\vvec{w}) \not\sim_{\mathcal{P}} G_{ik}(i,k,\vvec{v'},\vvec{w'})$.
 \end{enumerate}
 Suppose the condition~(1) holds.
 Then, w.l.o.g., $\bigcup_{x \in F_{ik}(X_i(\vvec{v'}))} [x]_{\sim_\mathcal{P}} \cap [(i,k,\vvec{v},\vvec{w})]_{\sim_\mathcal{P}} = \emptyset$
 for some $(i,k,\vvec{v},\vvec{w}){\hspace*{0.3pt}\in}$ $F_{ik}(X_i(\vvec{v}))$ by the definition of feasible relation.
 In particular, because $\bigcup_{x \in F_{ik}(X_i(\vvec{v'}))} [x]_{\sim_\mathcal{P}} \supseteq F_{ik}(X_i(\vvec{v'}))$,
 $F_{ik}(X_i(\vvec{v'})) \cap [(i,k,\vvec{v},\vvec{w})]_{\sim_\mathcal{P}} = \emptyset$.
 Let $\psi = [(i,k,\vvec{v},\vvec{w})]_{\sim_\mathcal{P}} \in \Psi_{ik}^\infty$.
 We have $F_{ik}(X_i(\vvec{v})) \cap \psi \supseteq \{(i,k,\vvec{v},\vvec{w})\} \not= \emptyset$ and $F_{ik}(X_i(\vvec{v'})) \cap \psi = \emptyset$.
 Therefore, $X_i(\vvec{v})$ and $X_i(\vvec{v'})$ belong different block when ${\otimes_{ik}^D}$ splits $\mathcal{P}$.
 This contradicts that $\mathcal{P}$ is a fixed-point of $H$.

 Moreover, suppose the condition~(2) holds.
 By a similar argument, there exists $X_a(d_a) \in G_{ik}(i,k,\vvec{v},\vvec{w})$ such that
 $G_{ik}(i,k,\vvec{v'},\vvec{w'}) \cap [X_a(d_a)]_{\sim_\mathcal{P}} = \emptyset$.
 Recall $G_{ik}(i,k,\vvec{v},\vvec{w}) = \{X_{a_1}(f_1(\vvec{v},\vvec{w})),\dots,X_{a_p}(f_p(\vvec{v},\vvec{w}))\}$,
 $a = a_q$ and $d_a = f_q(\vvec{v},\vvec{w})$ for some $1 \leq q \leq p$.
 Let $\phi = [X_a(d_a)]_{\sim_\mathcal{P}} \in \Phi_a$.
 Obviously, $G_{ik}(i,k,\vvec{v},\vvec{w}) \cap \phi \not= \emptyset$ and $G_{ik}(i,k,\vvec{v'},\vvec{w'}) \cap \phi = \emptyset$.
 Thus, $(i,k,\vvec{v},\vvec{w})$ and $(i,k,\vvec{v'},\vvec{w'})$ belong different block when ${\otimes_{ik}^B}$ splits $\mathcal{P}$.
 This contradicts $\mathcal{P}$ is a fixed-point.
 \QED
\end{proof}
\end{document}